\documentclass[11pt]{article}
\usepackage{fullpage}
\usepackage{color}
\usepackage{algorithm}
\usepackage{algorithmic}
\usepackage{amsmath}
\usepackage{amssymb}
\usepackage{amsthm}

\usepackage{color}
\usepackage{graphicx}
\usepackage{wrapfig}
\usepackage{subfigure}

\usepackage{epsfig}
\usepackage{clrscode}
\usepackage{epstopdf}

\def\compactify{\itemsep=0pt \topsep=0pt \partopsep=0pt \parsep=0pt}
\let\latexusecounter=\usecounter

{\begin{itemize}%
\setlength{\itemsep}{0pt}%
\setlength{\topsep}{0pt}%
\setlength{\partopsep}{0pt}%
\setlength{\parskip}{0pt}}%
{\end{itemize}}

\newcommand{\eps}{\epsilon}

\newtheorem{theorem}{Theorem}[section]

\newtheorem{lemma}[theorem]{Lemma}
\newtheorem{corollary}[theorem]{Corollary}

\newtheorem{definition}[theorem]{Definition}

\usepackage{fullpage}
\usepackage{hyperref}
\usepackage{enumerate}
\usepackage{appendix}

\usepackage{listings}

\lstdefinestyle{nonumbers}{numbers=none}
\lstdefinestyle{numbers}{numbers=left, numberstyle=\footnotesize, 
stepnumber=1, numbersep=7pt, 
xleftmargin=10pt}
\definecolor{gray}{RGB}{112,138,144}
\definecolor{orange}{RGB}{255,128,0}

\def\card#1{|#1|}

\def\dist{\mathrm{dist}}
\def\cost{\mathrm{Cost}}

\def\sol{\mathrm{Sol}}

{\makeatletter
 \gdef\xxxmark{%
   \expandafter\ifx\csname @mpargs\endcsname\relax 
     \expandafter\ifx\csname @captype\endcsname\relax 
       \marginpar{xxx}
     \else
       xxx 
     \fi
   \else
     xxx 
   \fi}
 \gdef\xxx{\@ifnextchar[\xxx@lab\xxx@nolab}
 \long\gdef\xxx@lab[#1]#2{{\bf [\xxxmark #2 ---{\sc #1}]}}
 \long\gdef\xxx@nolab#1{{\bf [\xxxmark #1]}}
 \long\gdef\xxx@lab[#1]#2{}\long\gdef\xxx@nolab#1{}%
}

\DeclareMathOperator*{\argmin}{arg\,min}

\usepackage{times}
\usepackage[compact]{titlesec}

\title{Simultaneous Nearest Neighbor Search\footnote{This work was in part supported by  NSF grant CCF 1447476~\cite{grant} and the Simons Foundation.}}
\author{Piotr Indyk \\MIT \\indyk@mit.edu \and Robert Kleinberg \\Cornell and MSR \\rdk@cs.cornell.edu \and Sepideh Mahabadi \\MIT \\mahabadi@mit.edu \and Yang Yuan \\Cornell University \\yy528@cornell.edu}
\date{}

\begin{document}

\maketitle

\begin{abstract}
Motivated by applications in computer vision and databases, we introduce and study the Simultaneous Nearest Neighbor Search (SNN) problem. Given a set of data points, the goal of SNN is to design a data structure that, given a {\em collection} of queries, finds a {\em collection} of close points that are ``compatible'' with each other. Formally, we are given $k$ query points $Q=q_1,\cdots,q_k$, and a compatibility graph $G$ with vertices in $Q$, and the goal is to return data points $p_1,\cdots,p_k$ that minimize (i) the weighted sum of the distances from $q_i$ to $p_i$ and (ii) the weighted sum, over all edges $(i,j)$ in the compatibility graph $G$, of the distances between $p_i$ and $p_j$. The problem has several applications in computer vision and databases, where one wants to return a set of {\em consistent} answers to multiple related queries. Furthermore, it generalizes several well-studied computational problems, including Nearest Neighbor Search, Aggregate Nearest Neighbor Search and the 0-extension problem.

In this paper we propose and analyze the following general two-step method for designing efficient data structures for SNN. In the first step, for each query point $q_i$ we find its (approximate) nearest neighbor point $\hat{p}_i$; this can be done efficiently using existing approximate nearest neighbor structures. In the second step, we solve an off-line optimization problem over sets $q_1,\cdots,q_k$ and $\hat{p}_1,\cdots,\hat{p}_k$; this can be done efficiently given that $k$ is much smaller than $n$. Even though $\hat{p}_1,\cdots,\hat{p}_k$  might not constitute the optimal answers to queries $q_1,\cdots,q_k$, we show that, for the unweighted case, the resulting algorithm satisfies a $O(\log k/\log \log k)$-approximation guarantee. Furthermore, we show that the approximation factor can be in fact reduced to a constant for compatibility graphs frequently occurring in practice, e.g., 2D grids, 3D grids or planar graphs. 

Finally, we validate our theoretical results by preliminary experiments. In particular, we show that the ``empirical approximation factor'' provided by the above approach is very close to 1.
 \end{abstract}

\section{Introduction}

The nearest neighbor search (NN) problem is defined as follows: given a collection $P$ of $n$ points, build a data structure that, given any query point from some set $Q$, reports the data point closest to the query. The problem is of key importance in many applied areas, including computer vision, databases, information retrieval, data mining, machine learning, and signal processing. The nearest neighbor search problem, as well as its approximate variants, have been a subject of extensive studies over the last few decades, see, e.g.,~\cite{bentley1975multidimensional, arya1998optimal, indyk1998approximate, kushilevitz2000efficient, krauthgamer2004navigating, andoni2014beyond} and the references therein.

Despite their success, however, the current algorithms suffer from significant theoretical and practical limitations. One of their major drawbacks is their inability to support  and exploit  {\em structure} in query sets that is often present in applications. Specifically, in many applications (notably in computer vision), queries issued to the data structure are not unrelated but instead correspond to samples taken from the same object.
For example, queries can correspond to pixels or small patches taken from the same image. To ensure consistency, one needs to impose ``compatibility constraints'' that ensure that related queries return similar answers.  Unfortunately, standard nearest neighbor data structures do not provide a clear way to enforce such constraints, as all queries are processed independently of each other. 

To address this issue, we introduce the {\em Simultaneous Nearest Neighbor Search} (SNN) problem. Given $k$ simultaneous query points $q_1,q_2,\cdots,q_k$, the goal of a SNN data structure is to find $k$ points (also called {\em labels}) $p_1,p_2,\cdots,p_k$ in $P$ such that (i) $p_i$ is close to $q_i$, and (ii) $p_1,\cdots,p_k$ are ``compatible''. 
Formally, the compatibility is defined by a graph $G=(Q,E)$ with $k$ vertices which is given to the data structure, along with the query points $Q=q_1,\cdots,q_k$. 
Furthermore,  we assume that the data set $P$ is a subset of some space $X$ equipped with a distance function $\dist_X$, and that we are given another metric $\dist_Y$ defined over $P \cup Q$.
Given the graph $G$ and the queries $q_1, \cdots ,q_k$, the goal of the SNN data structure is to return points $p_1, \cdots, p_k$  from $P$ that minimize the following function:
\begin{eqnarray}\label{eqn:costfunctiongeneral}
\sum_{i=1}^k \kappa_i \dist_Y(p_i,q_i) + \sum_{(i,j)\in E} \lambda_{i,j}\dist_X(p_i,p_j)
\end{eqnarray}
where $\kappa_i$ and $\lambda_{i,j}$ are parameters defined in advance. 

The above formulation captures a wide variety of applications that are not well modeled by traditional NN search. For example,  many applications in computer vision involve computing nearest neighbors of pixels or image patches from the same image~\cite{freeman2002example, boykov2001fast,barnes2009patchmatch}. In particular, algorithms for tasks such as de-noising (removing noise from an image),  restoration (replacing a deleted or occluded part of an image) or super-resolution (enhancing the resolution of an image) involve assigning ``labels'' to each image patch\footnote{This problem has been formalized in the algorithms literature as the {\em metric labeling} problem~\cite{kleinberg2002approximation}. 
The problem considered in this paper can thus be viewed as a variant of metric labeling with a very large number of labels.}. The labels could correspond to the pixel color, the enhanced image patch, etc. The label assignment should have the property that the labels are similar to the image patches they are assigned to, while at the same time the labels assigned to nearby image patches should be similar to each other.
The objective function in Equation~\ref{eqn:costfunctiongeneral} directly captures these constraints.

 From a theoretical perspective, Simultaneous Nearest Neighbor Search generalizes several well-studied computational problems, notably the Aggregate Nearest Neighbor problem~\cite{yiu2005aggregate, li2011flexible, li2011group, agarwal2012nearest,kopelowitz2012faster} and the $0$-extension problem~\cite{karzanov,FHRT,CKR,archer}. The first problem is quite similar to the basic nearest neighbor search problem over a metric $\dist$,  except that the data structure is given $k$ queries $q_1 \cdots q_k$, and the goal is to find a data point $p$ that minimizes the sum\footnote{Other aggregate functions, such as the maximum, are considered as well.} $\sum_i \dist(q_i,p)$. This objective can be easily simulated in SNN by setting $\dist_Y=\dist$ and 
 $\dist_X = L \cdot \mbox{uniform}$, where $L$ is a very large number and $\mbox{uniform}(p,q)$ is the uniform metric. The  0-extension problem is a combinatorial optimization problem where the goal is to minimize an objective function quite similar to that in Equation~\ref{eqn:costfunctiongeneral}. The exact definition of $0$-extension as well as its connections to SNN are discussed in detail in Section~\ref{ss:0ext}.

\subsection{Our results} In this paper we consider the basic case where $\dist_X=\dist_Y$ and $\lambda_{i,j}=\kappa_i=1$; we refer to this variant as the {\em unweighted} case. Our main contribution is a general reduction that enables us to design and analyze efficient data structures for unweighted SNN. The algorithm (called {\em Independent Nearest Neighbors} or {\em INN}) consists of two steps. In the first (pruning) step, for each query point $q_i$ we find its  nearest neighbor\footnote{Our analysis immediately extends to  the case where the we compute approximate, not exact, nearest neighbors. 
For simplicity we focus only on the exact case in the following discussion.}
 point $\hat{p}_i$ ; this can be done efficiently using existing nearest neighbor search data structures. In the second (optimization) step, we run an appropriate (approximation) algorithm for the SNN problem over sets $q_1,\cdots,q_k$ and $\hat{p}_1,\cdots,\hat{p}_k$; this can be done efficiently given that $k$ is much smaller than $n$. We show that the resulting algorithm satisfies a $O(b \log k/\log \log k)$-approximation guarantee, where $b$ is the approximation factor of the algorithm used in the second step.  
 This can be further improved to $O(b \delta )$, if the metric space $\dist$ admits a $\delta$-{\em padding decomposition}  (see Preliminaries for more detail). 
 The running time incurred by this algorithm is bounded by the cost of $k$  nearest neighbor search queries in a data set of size $n$ plus the  cost of the approximation algorithm for the $0$-extension problem over an input of size $k$. 
 By plugging in the best nearest neighbor algorithms for $\dist$ we obtain significant running time savings if $k \ll n$.

 We note that INN is somewhat similar to the belief propagation algorithm for super-resolution described in ~\cite{freeman2002example}. Specifically, that algorithm selects 16  closest labels for each $q_i$, and then chooses one of them by running a belief propagation algorithm that optimizes an objective function similar to Equation \ref{eqn:costfunctiongeneral}. However, we note that the algorithm in~\cite{freeman2002example} is heuristic and is not supported by approximation guarantees. 
 
 We complement our upper bound by showing that the aforementioned reduction inherently yields super-constant approximation guarantee. Specifically, we show that, for an appropriate distance function $\dist$, queries $q_1, \cdots ,q_k$, and  a label set $P$, the best solution to SNN with the label set restricted to $\hat{p}_1,\cdots,\hat{p}_k$ can be $\Theta(\sqrt{\log k})$ times larger than the best solution with label set equal to $P$. This means that even if the second step problem is solved to optimality, reducing the set of labels from $P$ to $\hat{P}$ inherently increases the cost by a super-constant factor.  

However, we further show that the aforementioned limitation can be overcome if the compatibility graph $G$ has pseudoarboricity $r$  (which means that each edge can be mapped to one of its endpoint vertices such that at most $r$ edges are mapped to each vertex). Specifically, we show that if $G$ has pseudoarboricity $r$, then the gap between the best solution using labels in $P$, and the best solution using labels in $\hat{P}$, is at most $O(r)$.
Since many graphs used in practice do in fact satisfy $r=O(1)$ (e.g., 2D grids, 3D grids or planar graphs), this means that the gap is indeed constant for a wide collection of common compatibility graphs.

In Appendix~\ref{s:2r1} we also present an alternative algorithm for the $r$-pseudoarboricity case. Similarly to INN, the algorithm computes the nearest label to each query $q_i$. However, the distance function used to compute the nearest neighbor involves not only the distance between $q_i$ and a label $p$, but also the distances between the {\em neighbors} of $q_i$ in $G$ and $p$. This nearest neighbor operation can be implemented using any data structure for the Aggregate Nearest Neighbor problem~\cite{yiu2005aggregate, li2011flexible, li2011group, agarwal2012nearest,kopelowitz2012faster}. Although this results in a more expensive query time, the labeling computed by this algorithm is final, i.e., there is no need for any additional postprocessing. Furthermore, the pruning gap (and therefore the final approximation ratio) of the algorithm is only $2r+1$, which is better than our bound for INN. 

Finally, we validate our theoretical results by preliminary experiments comparing our SNN data structure with an alternative (less efficient) algorithm that solves the same optimization problem using the full label set $P$. In our experiments we apply both algorithms to an image denoising task and measure their performance using the objective function~\eqref{eqn:costfunctiongeneral}. In particular, we show that the ``empirical gap'' incurred by the above approach, i.e,\ the ratio of objective function values observed in our experiments, is very close to 1.

\subsection{Our techniques}  We start by pointing out that SNN can be reduced to 0-extension in a ``black-box'' manner. Unfortunately,  this reduction yields an SNN algorithm whose running time depends on the size of labels $n$, which could be very large; essentially this approach defeats the goal of having a data structure solving the problem. The INN algorithm overcomes this issue by reducing the number of labels from $n$ to $k$. However the pruning step can increase the cost of the best solution. The ratio between the optimum cost after pruning to the optimum cost before pruning is called the {\em pruning gap}. 

To bound the pruning gap, we again resort to existing $0$-extension algorithms, albeit in a ``grey box'' manner. Specifically, we observe that many algorithms, such as those in \cite{CKR,archer, FHRT, leenaor}, proceed by first creating a label assignment in an ``extended'' metric space (using a LP relaxation of  $0$-extension),  and then apply a rounding algorithm to find an actual solution. The key observation is that the correctness of the rounding step does {\em not} rely on the fact that the initial label assignment is optimal, but instead it works for any label assignment. We use this fact to translate the known upper bounds for the integrality gap of linear programming relaxations of $0$-extension into upper bounds for the pruning gap. On the flip side, we show a lower bound for the pruning gap by mimicking the arguments used in~\cite{CKR} to lower bound the integrality gap  of a $0$-extension relaxation.

To overcome the lower bound, we consider the case where the compatibility graph $G$ has pseudoarboricity $r$. Many graphs used in applications, such as 2D grids, 3D grids or planar graphs, have pseudoarboricity $r$ for some constant $r$. We show that for such graphs the pruning gap is only $O(r)$. The proof proceeds by directly assigning labels in $\hat{P}$ to the nodes in $Q$ and bounding the resulting cost increase. 
It is worth noting that the ``grey box'' approach outlined in 
the preceding paragraph, combined with Theorem 11 of~\cite{CKR},
yields an $O(r^3)$ pruning gap for the class of $K_{r,r}$-minor-free
graphs, whose pseudoarboricity is $\tilde{O}(r)$. 
Our $O(r)$ pruning gap not only
improves this $O(r^3)$ bound in a quantitative
sense,
but it also applies to a much broader class of graphs. For example, 
three-dimensional grid graphs have pseudoarboricity 6, but the
class of three-dimensional grid graphs includes graphs with 
$K_{r,r}$ minors for every positive integer $r$. 

Finally, we validate our theoretical results by experiments. We focus on a simple de-noising scenario where $X$ is the pixel color space, i.e., the discrete three-dimensional space space $\{0 \ldots 255\}^3$. Each pixel in this space is parametrized by the intensity of the red, green and blue colors. We use the Euclidean norm to measure the distance between two pixels. We also let $P=X$. We consider three test images:  a cartoon with an MIT logo and two natural images. For each image we add some noise and then solve the SNN problems for both the full color space $P$ and the pruned color space $\hat{P}$. Note that since $P=X$, the set of pruned labels $\hat{P}$  simply contains all pixels present in the image. 

 Unfortunately, we cannot solve the problems optimally, since the best known exact algorithm takes exponential time. Instead, we run the same approximation algorithm on both instances and compare the solutions. We find that the values of the objective function for the solutions obtained using pruned labels and the full label space are equal up to a small multiplicative factor. This suggests that the empirical value of the pruning gap is very small, at least for the simple data sets that we considered.

\section{Definitions and Preliminaries}
We define the \emph{Unweighted Simultaneous Nearest Neighbor} problem as follows.
Let $(X,\dist)$ be a metric space and let $P\subseteq X$ be a set of $n$ points from the space.
 
\begin{definition}
In the \emph{Unweighted Simultaneous Nearest Neighbor} problem, the goal is to build a data structure over a given point set $P$ that supports the following operation. Given a set of $k$ points $Q=\{q_1, \cdots, q_k\}$ in the metric space $X$, along with a graph $G=(Q,E)$ of $k$ nodes, the goal is to report $k$ (not necessarily unique) points from the database $p_1, \cdots, p_k \in P$ which minimize the following cost function:
\begin{eqnarray}\label{eqn:costfunction}
\sum_{i=1}^k \dist(p_i,q_i) + \sum_{(q_i,q_j)\in E} \dist(p_i,p_j)
\end{eqnarray}
\noindent We refer to the first term in sum as the \emph{nearest neighbor (NN)} cost, and to the second sum as the \emph{pairwise (PW)} cost. We denote the cost of the optimal assignment from the point set $P$ by $\cost(Q, G, P)$.
\end{definition}
\noindent In the rest of this paper, simultaneous nearest neighbor (SNN)  refers to the unweighted version of the problem (unless stated otherwise).
Next, we define the \emph{pseudoarboricity} of a graph and \emph{$r$-sparse} graphs.

\begin{definition}
\emph{Pseudoarboricity} of a graph $G$ is defined to be the minimum number $r$, such that the edges of the graph can be oriented to form a directed graph with out-degree at most $r$. In this paper, we call such graphs as \emph{$r$-sparse}.
\end{definition}
\noindent Note that given an $r$-sparse graph, one can map the edges to one of its endpoint vertices such that there are at most $r$ edges mapped to each vertex. 
The doubling dimension of a metric space is defined as follows.
\begin{definition}\label{def:dd}
The \emph{doubling dimension} of a metric space $(X,\dist)$ is defined to be the smallest $\delta$ such that  every ball in $X$ can be covered by $2^\delta$ balls of half the radius. 
\end{definition}
\noindent It is known that the doubling dimension of any finite metric space is $O(\log \card{X})$. We then define padding decompositions.
\begin{definition}
A metric space $(X,\dist)$ is $\delta$-padded decomposable if for every $r$, there is a randomized partitioning of $X$ into clusters $\mathcal{C}=\{C_i\}$ such that, each $C_i$ has diameter at most $r$, and that for every $x_1,x_2\in X$, the probability that $x_1$ and $x_2$ are in different clusters is at most $\delta \dist(x_1,x_2)/r$.
\end{definition}
\noindent It is known that any finite metric with doubling dimension $\delta$ admits an $O(\delta)$-padding decomposition \cite{doublingdimension}.

\noindent 

\subsection{$0$-Extension Problem}
\label{ss:0ext}
The \emph{$0$-extension} problem, first defined by Karzanov \cite{karzanov} is closely related  to the Simultaneous Nearest Neighbor problem. In the $0$-extension problem, the input is a graph $G(V,E)$ with a weight function $w(e)$, and a set of terminals $T\subseteq V$ with a metric $d$ defined on $T$. The goal is to find a mapping from the vertices to the terminals $f: V\rightarrow T$ such that each terminal is mapped to itself and that the following cost function is minimized:
\[
\sum_{(u,v)\in E} w(u,v) \cdot d(f(u),f(v))
\]
It can be seen that this is a special case of the metric labeling problem \cite{kleinberg2002approximation} and thus a special case of the general version of the SNN problem defined by Equation \ref{eqn:costfunctiongeneral}. To see this, it is enough to let $Q=V$ and $P = T$, and let $\kappa_i = \infty$
for $q_i \in T$, $\kappa_i = 0$ for $q_i \not\in T$, 
and $\lambda_{i,j}=w(i,j)$ in Equation \ref{eqn:costfunctiongeneral}.

Calinescu et al. \cite{CKR} considered the semimetric relaxation of the LP for the $0$-extension problem and gave an $O(\log \card{T})$ algorithm using randomized rounding of the LP solution. They also proved an integrality ratio of $O(\sqrt{\log \card{T}})$ for the semimetric LP relaxation.

\noindent Later Fakcharoenphol et al. \cite{FHRT} improved the upper-bound to $O(\log \card{T} / \log \log \card{T})$, and Lee and Naor \cite{leenaor} proved that if the metric $d$ admits a $\delta$-padded decomposition, then there is an $O(\delta)$-approximation algorithm for the $0$-extension problem. For the finite metric spaces, this gives an $O(\delta)$ algorithm where $\delta$ is the doubling dimension of the metric space. Furthermore, the same results can be achieved using another metric relaxation (earth-mover relaxation), see \cite{archer}. Later Karloff et al. \cite{KKMR} proved that there is no polynomial time algorithm for $0$-extension problem with approximation factor $O((\log n)^{1/4-\eps})$ unless $NP\subseteq DTIME(n^{poly(\log n)})$.

SNN can be reduced to 0-extension in a ``black-box'' manner via the following lemma.
\begin{lemma}\label{lem:0reduction}
Any $b$-approximate algorithm for the $0$-extension problem yields an $O(b)$-approximate algorithm for the SNN problem.
\end{lemma}
\begin{proof}
Given an instance of the SNN problem $(Q,G',P)$, we build an instance of the $0$-extension problem $(V,T,G)$ as follows. Let $T=P$ and $V = T\cup Q$. The metric $d$ is the same as $\dist$. However the graph $G$ of the $0$-extension problem requires some modification. Let $G'=(Q,E_{G'})$, then $G=(V,E)$ is defined as follows. For each $q_i,q_j\in Q$, we have the edge $(q_i,q_j)\in E$ iff $(q_i,q_j)\in E_{G'}$.  We also include another type of edges in the graph: for each $q_i\in Q$, we add an edge $(q_i,\hat p_i)\in E$ where $\hat p_i \in P$ is the  nearest neighbor of $q_i$. Note that we consider the graph $G$ to be unweighted.

Using the $b$-approximation algorithm for this problem, we get an assignment $\mu$ that maps the non-terminal vertices $q_1, \cdots, q_k$  to the terminal vertices. Suppose $q_i$ is mapped to the terminal vertex $p_i$ in this assignment. Let $p_1^*,\cdots, p_k^*$ be the optimal SNN assignment.
 Next, we show that  the same mapping $\mu$ for the SNN problem, gives us an $O(b)$ approximate solution. The SNN cost of the mapping $\mu$ is denoted as follows:
\[
\begin{split}
\cost^{\mathrm{SNN}}(\mu) &= \sum_{i=1}^k \dist(q_i , p_i) + \sum_{(q_i,q_j)\in E_{G'}} \dist(p_i , p_j)\\
&\leq \sum_{i=1}^k \dist(q_i , \hat p_i) + \sum_{i=1}^k \dist(\hat p_i , p_i) + \sum_{(q_i,q_j)\in E_{G'}} \dist(p_i , p_j) \\
&\leq \sum_{i=1}^k \dist(q_i , p_i^*) + b \cdot [\sum_{i=1}^k \dist(\hat p_i , p_i^*) + \sum_{(q_i,q_j)\in E_{G'}} \dist(p_i^* , p_j^*)] \\
&\leq \cost(Q,G',P) + b\cdot [\sum_{i=1}^k \dist(\hat p_i , q_i) + \sum_{i=1}^k \dist(q_i , p_i^*) +\sum_{(q_i,q_j)\in E_{G'}} \dist(p_i^* , p_j^*)]\\
&\leq \cost(Q,G',P) + b\cdot [\sum_{i=1}^k \dist(\hat p_i , q_i) + \cost(Q,G',P)]\\
&\leq \cost(Q,G',P) (2b +1)
\end{split}
\]
\noindent where we have used triangle inequality and the following facts in the above. First,  $\hat p_i$ is the closest point in $P$ to $q_i$ and thus $\dist(q_i,\hat p_i)\leq \dist(q_i,p_i^*)$. Second,  by definition we have that $\cost(Q,G',P) = \sum_{i=1}^k \dist(q_i , p_i^*) + \sum_{(q_i,q_j)\in E_{G'}} \dist(p_i^* , p_j^*)$. Finally, since $\mu$ is a $b$ approximate solution for the 0-extension problem, we have that $\sum_{i=1}^k \dist(\hat p_i , p_i) + \sum_{(q_i,q_j)\in E_{G'}} \dist(p_i , p_j)$ is smaller than $b$ times the 0-extension cost of any other assignment, and in particular $\sum_{i=1}^k \dist(\hat p_i , p_i^*) + \sum_{(q_i,q_j)\in E_{G'}} \dist(p_i^* , p_j^*)$.
\end{proof}
By plugging in the known $0$-extension algorithms cited earlier we obtain the following:
\begin{corollary}
There exists an $O(\log n/\log \log n)$ approximation algorithm for the SNN problem with running time $n^{O(1)}$, where $n$ is the size of the label set.
\end{corollary}

\begin{corollary}
If the metric space $(X,\dist)$ is $\delta$-padded decomposable, then there exists an $O(\delta)$ approximation algorithm for the SNN problem  with running time $n^{O(1)}$. For finite metric spaces $X$, $\delta$ could represent the doubling dimension of the  metric space (or equivalently the doubling dimension of $P\cup Q$).
\end{corollary}

Unfortunately, this reduction yields a SNN algorithm with running time depending on the size of labels $n$, which could be very large.  
In the next section we show how to improve the running time by reducing the labels set size from $n$ to $k$.
However, unlike the reduction in this section, our new reduction will no longer be ``black-box''.
Instead, its analysis will use {\em particular properties} of the 0-extension algorithms. 
Fortunately those properties are satisfied by the known approximation algorithms for this problem. 


\section{Independent Nearest Neighbors Algorithm}
In this section, we consider a natural and general algorithm for the SNN problem, which we call  {\em Independent Nearest Neighbors (INN)}.
The algorithm  proceeds as follows. Given the query points $Q=\{q_1,\cdots,q_k\}$, for each $q_i$ the algorithm picks its (approximate) nearest neighbor $\hat{p_i}$. Then it solves the problem over the set $\hat{P}  = \{\hat{p_1},\cdots,\hat{p_k}\}$ instead of $P$. This simple approach reduces the size of search space from $n$ down to $k$.

The details of the algorithm are shown in Algorithm \ref{alg:inn}.

\begin{algorithm}[!h]
\caption{Independent Nearest Neighbors (INN) Algorithm}
\label{alg:inn}
\vspace{0.2cm}
\textbf{Input} $Q=\{q_1,\cdots,q_k\}$, and input graph $G=(Q,E)$\\
\vspace{-0.5cm}
\begin{algorithmic}[1]

\vspace{0.2cm}
\FOR{$i=1$ \TO $k$}
	\STATE{Query the NN data structure to extract a nearest neighbor (or approximate nearest neighbor) $\hat{p_i}$ for $q_i$}
\ENDFOR
\STATE{Find the optimal (or approximately optimal) solution among the set $\hat{P} = \{\hat{p_1},\cdots,\hat{p_k}\}$.}
\end{algorithmic}
\end{algorithm}

In the rest of the section we analyze the quality of this pruning step. More specifically, we define the \emph{pruning gap} of the algorithm as the ratio of the optimal cost function using the points in $\hat P$ over its value using the original point set $P$.
\begin{definition}
The \emph{pruning gap} of an instance of SNN is defined as $\alpha(Q,G,P) = {\cost(Q,G,\hat{P}) \over \cost(Q,G,P)}$. We define the pruning gap of the INN algorithm, $\alpha$, as the largest value of $\alpha(Q,G,P)$ over all instances.
\end{definition}

First, in Section \ref{sec:general}, by proving a reduction from algorithms for rounding the LP solution of the $0$-extension problem, we show that for arbitrary graphs $G$, we have $\alpha = O(\log k/\log \log k)$, and if the metric $(X,\dist)$ is $\delta$-padded decomposable, we have $\alpha = O(\delta)$ (for example, for finite metric spaces $X$, $\delta$ can represent the doubling dimension of the metric space). Then, in Section \ref{sec:sparse}, we prove that $\alpha=O(r)$ where $r$ is the pseudoarboricity of the graph $G$. This would show that for the sparse graphs, the pruning gap remains constant. Finally, in Section \ref{sec:lowerbound}, we present a lower bound showing that the pruning gap could be as large as $\Omega(\sqrt{\log k})$ and as large as $\Omega(r)$ for ($r\leq \sqrt{\log k}$).  Therefore, we get the following theorem.

\begin{theorem}
The following bounds hold for the pruning gap of the INN algorithm. First we have $\alpha = O({\log k\over\log \log k})$, and that if metric $(X,\dist)$ is $\delta$-padded decomposable, we have $\alpha =O(\delta)$. Second, $\alpha = O(r)$ where $r$ is the pseudoarboricity of the graph $G$. Finally, we have that $\alpha = \Omega(\sqrt {\log k})$ and $\alpha = \Omega(r)$ for $r\leq \sqrt{\log k}$.
\end{theorem}

Note that the above theorem results in an $O(b\cdot \alpha)$ time algorithm for the SNN problem where $b$ is the approximation factor of the algorithm used to solve the metric labeling problem for the set $\hat P$, as noted in line 4 of the INN algorithm. For example in a general graph $b$ would be $O(\log k/\log \log k)$ that is added on top of $O(\alpha)$ approximation of the pruning step.


\subsection{Bounding the pruning gap using $0$-extension}\label{sec:general}
In this section we show upper bounds for the pruning gap ($\alpha$) of the INN algorithm. The proofs use specific properties of existing algorithms for the $0$-extension problem.
\begin{definition}\label{def:rounding-alg}
We say an algorithm $A$ for the $0$-extension problem is a \emph{$\beta$-natural rounding} algorithm if, given a graph $G=(V,E)$, a set of terminals $T\subseteq V$, a metric space $(X,d_X)$, and a mapping $\mu: V\rightarrow X$, it outputs another mapping $\nu:V\rightarrow X$ with the following properties:
\begin{itemize}
\item $\forall t\in T: \nu(t)=\mu(t)$
\item  $\forall v\in V: \exists t\in T \mbox{\ \ s.t.\ \ }  \nu(v) = \mu(t)$
\item $\cost(\nu)\leq \beta \cost(\mu)$, i.e., $\sum_{(u,v)\in E} d_X(\nu(u),\nu(v)) \leq  \beta\cdot \sum_{(u,v)\in E} d_X(\mu(u),\mu(v))$
\end{itemize}
\end{definition}

\noindent Many previous algorithms for the $0$-extension problem, such as \cite{CKR,archer, FHRT, leenaor}, first create the mapping $\mu$ using some LP relaxation of  $0$-extension (such as semimetric relaxation or earth-mover relaxation), and then apply a $\beta$-natural rounding algorithm for the $0$-extension to find the mapping  $\nu$ which yields the solution to the $0$-extension problem.
Below we give a formal connection between guarantees of these rounding algorithms, and the quality of the output of the INN algorithm (the pruning gap of INN). 

\begin{lemma}
Let $A$ be a $\beta$-natural rounding algorithm for the $0$-extension problem. Then we can infer that the pruning gap of the INN algorithm is $O(\beta)$, that is, $\alpha = O(\beta)$. 
\end{lemma}

\begin{proof}
Fix any SNN instance $(Q,G_S,P)$, where $G_S=(Q,E_{PW})$, and its corresponding INN invocation. 

We construct the inputs to the algorithm $A$ from the INN instance as follows. Let the metric space of $A$ be the same as $(X,\dist)$ defined in the SNN instance. Also, let $V$ be a set of $2k$ vertices corresponding to $\hat{P} \cup P^*$ with $T$ corresponding to $\hat{P}$. Here $P^*=\{p_1^*,\cdots,p_k^*\}$ is the set of the optimal solutions of SNN, and $\hat P$ is the set of nearest neighbors as defined by INN. The mapping $\mu$ simply maps each vertex from $V=\hat{P} \cup P^*$ to itself in the metric $X$ defined in SNN.
Moreover, the graph $G=(V,E)$ is defined such that $E=\{(\hat{p}_i, p_i^*)| 1\leq i\leq k\} \cup 
\{(p_i^*, p_j^*)| (q_i, q_j)\in E_{PW}\}$.

First we claim the following (note that $\cost(\mu)$ is defined in Definition \ref{def:rounding-alg}, and that by definition $\cost(Q,G_S,P) =\cost(Q,G_S,P^*) $)
\[
\cost(\mu)\leq 2\cost(Q,G_S,P^*) = 2\cost(Q,G_S,P)
\]
We know that $\cost(Q,G_S,P^*)$ can be split into NN cost and PW cost. We can also split $\cost(\mu)$ into NN cost (corresponding to edge set $\{(\hat{p}_i, p_i^*)| 1\leq i\leq k\}$) and PW cost (corresponding to edge set $\{(p_i^*, p_j^*)| (q_i, q_j)\in E_{PW}\}$).
By definition we know the PW costs of $\cost (Q,G_S,P)$
and $\cost(\mu)$ are equal. For NN cost, by triangle inequality, we know 
$\dist(\hat p_i, p_i^*)\leq \dist(\hat p_i, q_i)+\dist(q_i, p_i^*)\leq 2\cdot \dist(q_i, p_i^*) $. Here we use the fact that $\hat p_i$ is the nearest database point of $q_i$. Thus, the claim follows. 

We then apply algorithm $A$ to get the mapping $\nu$. By the assumption on $A$, we know that $\cost(\nu)\leq \beta\cost(\mu)$.  Given the mapping $\nu$ by the algorithm $A$, consider the assignment in the SNN instance where each query $q_i$ is mapped to $\nu(p_i^*)$, and note that since $\nu(p_i^*) \in T$, this would map all points $q_i$  to points in $\hat P$.
Thus, by definition, we have that
\[
\begin{split}
\cost(Q,G_S,\hat P)&\leq \sum_{i=1}^k \dist(q_i, \nu(p_i^*)) + \sum_{(q_i,q_j)\in E_{PW}} \dist(\nu(p_i^*),\nu(p_j^*)) \\
&\leq \sum_{i=1}^k \dist(q_i,\hat p_i)+\sum_{i=1}^k \dist(\hat p_i, \nu(p_i^*)) + \sum_{(q_i,q_j)\in E_{PW}} \dist(\nu(p_i^*),\nu(p_j^*)) \\
&\leq \sum_{i=1}^k \dist(q_i,\hat p_i) +\cost(\nu)  \\
&\leq \cost(Q,G_S , P)+\beta\cost(\mu)  \\
&\leq (2\beta+1)\cost(Q,G_S , P)
\end{split}
\]
\noindent where we have used the triangle inequality. Therefore, we have that the pruning gap $\alpha$ of the INN algorithm is $O(\beta)$, as claimed. 
\end{proof}

Using the previously cited results, and noting that in the above instance $\card{V}=O(k)$, we get the following corollaries.
\begin{corollary}
The INN algorithm has pruning gap $\alpha = O(\log k/\log \log k)$.
\end{corollary}
\begin{corollary}
If the metric space $(X,\dist)$ admits a $\delta$-padding decomposition, then the INN algorithm has pruning gap $\alpha = O(\delta)$. For finite metric spaces $(X,\dist)$, $\delta$ is at most the doubling dimension of the metric space.
\end{corollary}


\subsection{Sparse Graphs}\label{sec:sparse}
In this section, we prove that the INN algorithm performs well on sparse graphs. More specifically, here we prove that when the graph $G$ is $r$-sparse, then $\alpha(Q,G,P) = O(r)$. To this end, we show that there exists an assignment using the points in $\hat{P}$ whose cost function is within $O(r)$ of the optimal solution using the points in the original data set $P$.

Given a graph $G$ of pseudoarboricity $r$, we know that we can map each edge to one of its end points such that the number of edges mapped to each vertex is at most $r$. For each edge $e$, we call the vertex that $e$ is mapped to as the \emph{corresponding vertex} of $e$. This would mean that each vertex is the corresponding vertex of at most $r$ edges.

Let $p_1^*,\cdots,p_k^* \in P$ denote the optimal solution of SNN. Algorithm \ref{alg:inn:degenerate} shows how to find an assignment $p_1,\cdots,p_k \in \hat{P}$. We show that the cost of this assignment is within a factor $O(r)$ from the optimum.

\begin{algorithm}[!h]
\caption{$r$-Sparse Graph Assignment Algorithm}
\label{alg:inn:degenerate}
\vspace{0.2cm}
\textbf{Input} Query points $q_1,\cdots,q_k$, Optimal assignment $p_1^*,\cdots,p_k^*$, Nearest Neighbors $\hat{p_1},\cdots,\hat{p_k}$, and the input graph $G=(Q,E)$\\
\textbf{Output} An Assignment $p_1,\cdots,p_k \in \hat{P}$ \\
\vspace{-0.5cm}
\begin{algorithmic}[1]

\vspace{0.2cm}
\FOR{$i=1$ \TO $k$}
	\STATE{Let $j_0 = i$ and let $q_{j_1},\cdots,q_{j_t}$ be all the neighbors of $q_i$ in the graph $G$}
	\STATE{ $m \leftarrow \argmin_{\ell=0}^t \dist(p_i^*,p_{j_\ell}^*) + \dist(p_{j_\ell}^* , q_{j_\ell})$}
	\STATE{Assign $p_i \leftarrow \hat{p}_{j_m}$}
\ENDFOR
\end{algorithmic}
\end{algorithm}

\begin{lemma}\label{lem:degenerate}
The assignment defined by Algorithm \ref{alg:inn:degenerate}, has $O(r)$ approximation factor.
\end{lemma}
\begin{proof}
For each $q_i\in Q$, let $y_i = \dist(p_i^* , q_i)$ and for each edge $e=(q_i,q_j) \in E$ let $x_e = \dist(p_i^* , p_j^*)$. 
Also let $Y=\sum_{i=1}^k y_i$ and $X=\sum_{e\in E} x_e$. Note that $Y$ is the NN cost and $X$ is the PW cost of the optimal assignment and that $OPT = \cost(Q,G,P) = X +Y$. 
Define the variables $y_i\rq{}$, $x_e\rq{}$, $Y\rq{}$ , $X\rq{}$ in the same way but for the assignment $p_1,\cdots,p_k$ produced by the algorithm. 
That is, for each $q_i\in Q$, $y_i'= \dist(p_i , q_i)$, and for each edge $e=(q_i,q_j) \in E$, $x_e' = \dist(p_i , p_j)$. 
Moreover, for a vertex $q_i$, we define the \textit{designated neighbor} of $q_i$ to be $q_{j_m}$ for the value of $m$ defined in the line 3 of Algorithm \ref{alg:inn:degenerate} (note that the designated neighbor might be the vertex itself).
Fix a vertex $q_i$ and let $q_c$ be the designated neighbor of $q_i$. We can bound the value of $y_i\rq{}$ as follows.
\[
\begin{split}
y_i\rq{} = \dist(q_i,p_i) &= \dist(q_i,\hat{p}_{c}) \\
&\leq \dist(q_i,p_i^*) + \dist(p_i^*,p_c^*) + \dist(p_c^* , q_c) + \dist(q_c,\hat{p}_c) \quad \mbox{(by triangle inequality)} \\
&\leq y_i + \dist(p_i^*,p_c^*) + 2\dist(p_c^*,q_c) \quad \mbox{(since $\hat{p_c}$ is the nearest neighbor of $q_c$)} \\
&\leq y_i + 2[\dist(p_i^*,p_c^*) + \dist(p_c^*,q_c)] \\
&\leq 3y_i \quad \mbox{(by definition of  designated neighbor and the value $m$ in line 3 of Algorithm \ref{alg:inn:degenerate})}\\
\end{split}
\]

\noindent Thus summing over all vertices, we get that $Y'\leq 3Y$. Now for any fixed edge $e = (q_i,q_s)$ (with $q_i$ being its corresponding vertex), let $q_c$ be the designated neighbor of $q_i$, and $q_z$ be the designated neighbor of $q_s$. 
Then we bound the value of $x_e\rq{}$ as follows.
\[
\begin{split}
x_e\rq{} &= \dist(p_i,p_s) = \dist(\hat{p}_c,\hat{p}_{z}) \quad \mbox{(by definition of designated neighbor and line 4 of Algorithm \ref{alg:inn:degenerate})} \\\
&\leq \dist(\hat{p}_c,q_c) + \dist(q_c,p_c^*) + \dist(p_c^*,p_i^*) + \dist(p_i^*,p_s^*) \\
&+ \dist(p_s^*,p_{z}^*) + \dist(p_{z}^*,q_{z})+ \dist(q_{z},\hat{p}_{z}) \quad \mbox{(by triangle inequality)}\\
&\leq 2 \dist(q_c,p_c^*) + \dist(p_c^*,p_i^*) + \dist(p_i^*,p_s^*) \\
&+ \dist(p_s^*,p_{z}^*) + 2\dist(p_{z}^*,q_{z}) \quad \mbox{(since $\hat{p}_c (\hat{p}_{z}\mbox{ respectively})$ is  a NN of $q_c(q_{z}\mbox{ respectively})$)} \\
&\leq 2 [\dist(q_c,p_c^*) + \dist(p_c^*,p_i^*)] + \dist(p_i^*,p_s^*) + 2[\dist(p_s^*,p_{z}^*) + \dist(p_{z}^*,q_{z})] \\
&\leq 2 y_i + x_e+ 2 [x_e+y_i] \quad \mbox{(since $q_c (q_z \mbox{ respectively}) $ is designated neighbor of $q_i(q_s \mbox{ respectively})$)} \\
&\leq 4(x_e+y_i)
\end{split}
\]
\noindent Hence, summing over all the edges, since each vertex $q_i$ is the corresponding vertex of at most $r$ edges, we get that $X'\leq 4X + 4rY$. Therefore we have the following.
\[
\cost(Q,G,\hat{P}) \leq X\rq{}+Y\rq{} \leq 3Y + 4X + 4rY \leq (4r+3)\cdot \cost(Q,G,P)
\]
\noindent and thus $\alpha(Q,G,P) = O(r)$.
\end{proof}

\section{Lower bound}\label{sec:lowerbound}
In this section we prove a lower bound of $\Omega(\sqrt{\log k})$ for the approximation factor of the INN algorithm. Furthermore, the lower bound example presented in this section is a graph (in fact a multi-graph) that has pseudoarboricity equal to $O(\sqrt {\log k})$, showing that in a way, the upper bound of $\alpha = O(r)$ for the $r$-sparse graphs is tight. More specifically, we show that for $r\leq \sqrt{\log k}$, we have $\alpha = \Omega(r)$.
We note that the lower bound construction presented in this paper is similar to the approach of \cite{CKR} for proving a lower bound for the integrality ratio of the LP relaxation for the $0$-extension problem. 

\begin{lemma}
For any value of $k$, there exists a set of points $P$ of size $O(k)$ in a metric space $X$, and a query $(Q,G)$ such that $\card{Q}=k$ and the pruning step induces an approximation factor of at least $\alpha(Q,G,P) = \Omega(\sqrt{\log k})$.
\end{lemma}

\begin{proof}
\noindent In what follows, we describe the construction of the lower bound example.

Let $H=(V,E)$ be an expander graph with $k$ vertices $V=\{v_1,\cdots, v_k\}$ such that each vertex has constant degree $d$ and the vertex expansion of the graph is also a constant $c$. Let $H'=(V',E',W')$ be a weighted graph constructed from $H$ by adding $k$ vertices $\{u_1,\cdots, u_k\}$ such that each new vertex $u_i$  is a leaf connected to $v_i$ with an edge of weight $\sqrt {\log k}$. All the other edges between $\{v_1,\cdots, v_k\}$ (which were present in $H$) have weight $1$. This graph $H'$ defines the metric space $(X,\dist)$ such that $X$ is the set of nodes $V'$ and $\dist$ is the weight of the shortest path between the nodes in the graph $H'$.  Moreover, let $P=V'$ be all the vertices in the graph $H'$.

\noindent Let the set of $k$ queries be $Q = V'\setminus V = \{u_1,\ldots,u_k\}$.
Then, while running the INN algorithm, the set of candidates $\hat{P}$ would be the queries themselves, i.e., $\hat{P} = Q = \{u_1,\cdots,u_k\}$. Also, let the input graph $G=(Q,E_G)$ be a multi-graph which is obtained from $H$ by 
replacing each edge $(v_i,v_j)$ in $H$ with $\sqrt{\log k}$ copies of the
edge $(u_i,u_j)$ in $G$. 
This is the input graph given along with the $k$ queries to the algorithm.

\noindent Consider the solution $P^* = \{p_1^*, \cdots, p_k^*\}$ where $p_i^* = v_i$.  The cost of this solution is 
\[
\sum_{i=1}^k \dist(q_i,p_i^*)+\sum_{(u_i,u_j)\in E_G}\dist(v_i,v_j) = k\sqrt{\log k} + kd\sqrt {\log k} /2
\]
Therefore, the cost of the optimal solution $OPT=\cost(Q,G,P)$ is at most $O(k\sqrt {\log k})$. Next, consider the optimal labeling $\hat{P}^*=\{\hat{p_1^*}, \cdots, \hat{p_k^*}\}\subseteq \hat P$ using only the points in $\hat{P}$. This optimal assignment has one of the following forms.

\noindent \textbf{Case 1:} For all $1\leq i \leq k$, we have $\hat{p_i^*} = u_i$. The cost of $\hat{P}^*$ in this case would be
\[
\cost(Q,G,\hat{P})  = \sum_{i=1}^k \dist(q_i,u_i) + \sum_{(u_i,u_j)\in E_G} \dist(u_i,u_j) \geq 0 + |E_G| \cdot 2\sqrt{\log k} \geq {d k\over2} \log k
\]
\noindent Thus the cost in this case would be $\Omega(OPT\sqrt {\log k})$.

\noindent \textbf{Case 2:} All the $\hat{p_i^*}$ 's are equal. Without loss of generality suppose they are all equal to $u_1$. Then the cost would be:
\[
\cost(Q,G,\hat{P}) = \sum_{i=1}^k \dist(q_i,u_1) + \sum_{(u_i,u_j)\in E_G} \dist(u_1,u_1) \geq \Omega(k\log k) +0 
\]
\noindent This is true because in an expander graph with constant degree, the number of vertices at distance less than $\log_d k \over 2$ of any vertex is at most $1+d+\cdots, d^{\log_d k \over 2} \leq 2\sqrt k$. Thus $\Theta(k)$ vertices are farther than ${\log_d k \over 2} = {\log k \over 2\log d} = \Theta(\log k)$.
Thus, again the cost of the assignment $\hat{P}$ in this case would be $\Omega(OPT \sqrt{\log k})$.

\noindent \textbf{Case 3:} Let $S=\{S_1,\cdots, S_t\}$ be a partition of $[k]$ such that each part corresponds to all the indices $i$ having their $\hat{p_i^*}$ equal. That is, for each $1\leq j \leq t$, we have $\forall i,i' \in S_j: \hat{p_i^*} = \hat{p_{i'}^*}$.  Now, two cases are possible.
First if all the parts $S_j$ have size at most $k/2$. In this case, since the graph $H$ has expansion $c$, the total number of edges between different parts would be at least  
\[
\left|\{(u_i,u_j)\in E_G | \hat{p_i^*} \neq \hat{p_j^*}\}\right|\geq {1\over 2}\sum_{j=1}^t c|S_j|\sqrt{\log k} \geq kc\sqrt{\log k}/2
\]
\noindent Therefore similar to Case 1 above, the PW cost would be at least $kc\sqrt{\log k}/2 \cdot \sqrt{\log k} = \Omega(k\log k)$.
Otherwise, at least one of the parts such as $S_j$ has size at least $k/2$. In this case, similar to Case 2 above, the NN cost would be at least $\Omega (k\log k)$. Therefore, in both cases the cost of the assignment $\hat{P}^*$ would be at least $\Omega(OPT \sqrt{\log k})$. Hence, the pruning gap of the INN algorithm on this graph is $\Omega(\sqrt{\log k})$.
\end{proof}

Since the degree of all the vertices in the above graph is $d\sqrt {\log k}$, the pseudoarboricity of the graph is  also $\Theta(\sqrt {\log k})$. It is easy to check that if we repeat each edge $r$ times instead of $\sqrt{\log k}$ times in $E_G$ in the above proof, the same arguments hold and we get the following corollary.

\begin{corollary}
For any value of $r\leq \sqrt{\log k}$, there exists an instance of SNN(Q,G,P) such that the input graph $G$ has arboricity $O(r)$ and that the pruning gap of the INN algorithm is $\alpha(Q,G,P)=\Omega(r)$.
\end{corollary}

\section{Experiments}
We consider image denoising as an application of our algorithm. A popular approach to denoising (see e.g. \cite{treemetric}) is to minimize the following objective function:
\[
 \sum_{i\in V} \kappa_i d(q_i, p_i) +  \sum_{(i,j)\in E}\lambda_{i,j} d(p_i,p_j)
\]
\noindent Here $q_i$ is the color of pixel $i$ in the noisy image, and $p_i$ is the color of pixel $i$ in the output.
We use the standard 4-connected neighborhood system for the edge set $E$, and use Euclidean distance as the distance function $d(\cdot,\cdot)$. We also set all weights $\kappa_i$ and $\lambda_{i,j}$ to 1.

When the image is in grey scale, this objective function can be optimized approximately and efficiently using message passing algorithm, see e.g. \cite{earlyvision}. However, when the image pixels are points in RGB color space, the label set becomes huge ($n=256^3=16,777,216$), and most techniques for metric labeling are not feasible.  

Recall that our algorithm proceeds by considering only the nearest neighbor labels of the query points,  i.e., only  the colors that appeared in the image. In what follows we refer to this reduced set of labels as the {\em image color} space, as opposed to  the {\em full color} space where no pruning is performed.  

In order to optimize the objective function efficiently, 
we use the technique of \cite{treemetric}.
We first embed the original (color) metric space into a tree metric (with $O(\log n)$ distortion), and then apply a top-down divide and conquer algorithm on the tree metric, by calling the alpha-beta swap subroutine \cite{gco}. 
%
 We use the random-split kd-tree for both the full color space and the image color space. When constructing the kd-tree, split each interval  $[a,b]$ by selecting a random number chosen uniformly at random from the interval $[0.6a+0.4b,0.4a+0.6b]$.

To evaluate the performance of the two algorithms, we use one cartoon image with MIT logo and  two images from the Berkeley segmentation dataset \cite{berkeleyseg} which was previously used in other computer vision papers \cite{treemetric}. 
We use Matlab imnoise function to create noisy images from the original images.
We run each instance  $20$ times, and compute both the average and the variance of the objective function
(the variance is due to the random generating process of kd tree).

\begin{table}
\centering
 \begin{tabular}{ c | c | c |c}
   \hline
   & Avg cost for full color & Avg cost for image color&
   Empirical pruning gap\\
   \hline
   MIT & $341878 \pm 3.1\%$ & $340477 \pm 1.1$\% &0.996\\ \hline
   Snow & $9338604 \pm 4.5\%$ & $9564288 \pm 6.2\%$ &1.024\\ \hline
   Surf & $8304184 \pm 6.6\%$ & $7588244 \pm 5.1\%$ &0.914\\
   \hline
 \end{tabular}
\caption{The empirical values of objective functions for the respective images and algorithms.}
\label{tab:costfunc}
\end{table}

\begin{table}[!h]
\centering
 \begin{tabular}{ccc}

  \includegraphics[width=43mm]{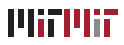}
&
  \includegraphics[width=43mm]{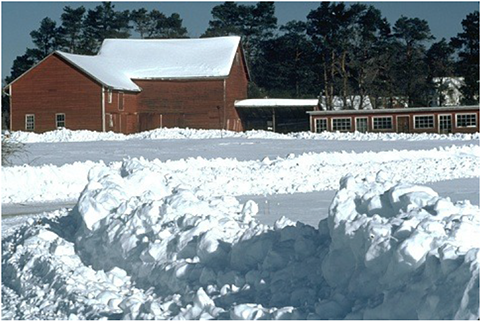}
&
  \includegraphics[width=43mm]{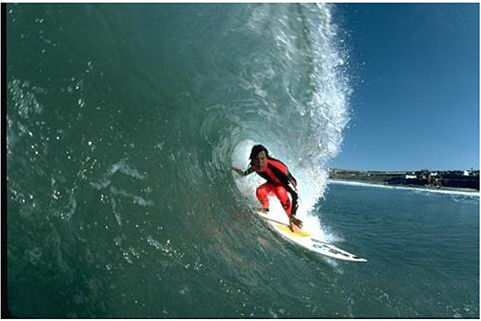}
\\
  \includegraphics[width=45mm]{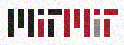}
&
  \includegraphics[width=45mm]{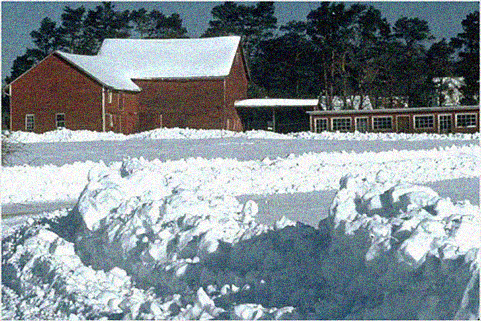}
&
  \includegraphics[width=45mm]{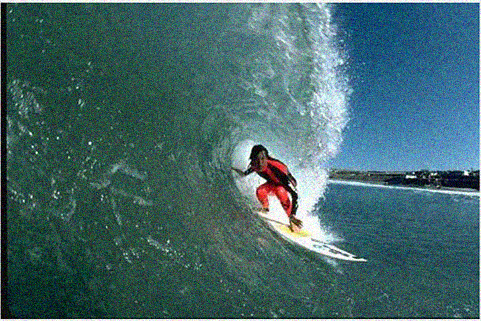}
\\
  \includegraphics[width=45mm]{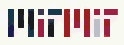}
&
  \includegraphics[width=45mm]{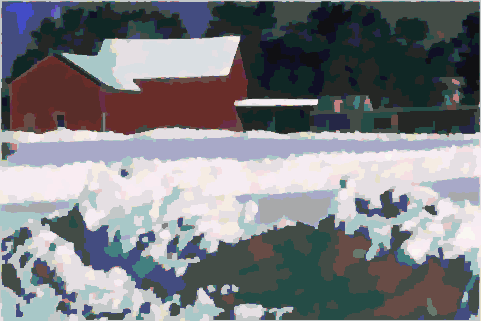}
&
  \includegraphics[width=45mm]{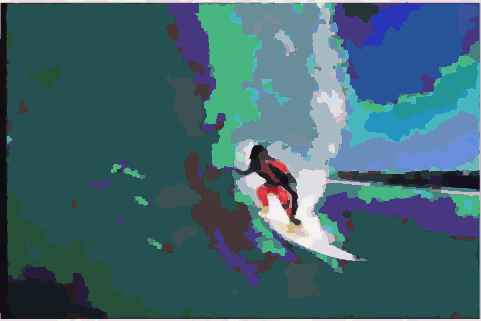}
\\
    \includegraphics[width=45mm]{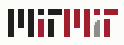}
&
  \includegraphics[width=45mm]{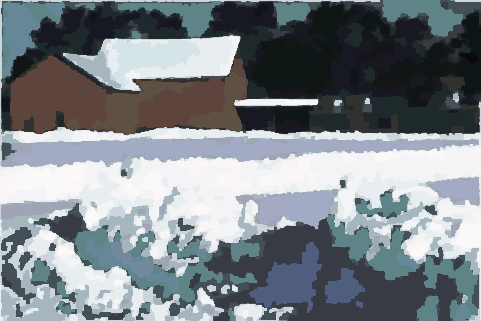}
&
   \includegraphics[width=45mm]{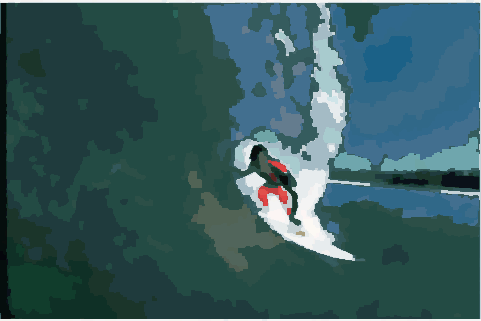}
\\
\end{tabular}
\caption{MIT logo (first column, size $45*124$), and two images from the Berkeley segmentation dataset \cite{berkeleyseg} (second \& third columns, size $321*481$). The first row shows the original image; the second row shows the noisy image; the third row shows the denoised image using full color space; the fourth row shows the denoised image using image space (our algorithm).
}
\label{fig:results}
\end{table}

The results are presented in Figure \ref{fig:results} and Table \ref{tab:costfunc}. 
In Figure \ref{fig:results}, one can see that the images produced by the two algorithms are comparable.
The full color version seems to preserve a few more details than the image color version, but it also ``hallucinates'' non-existing colors to minimize the value of the objective function.
The visual quality of the de-noised images can be improved by fine-tuning various parameters of the algorithms. 
We do not report these results here, as our goal  was to compare the values of the objective function produced by the two algorithms, as opposed to developing the state of the art de-noising system. 

Note that, as per Table \ref{tab:costfunc}, for some images the value of the objective function is sometimes {\em lower} for the image color space compared to the full color space. This is because we cannot solve the optimization problem exactly.  In particular, using the kd tree to embed the original metric space into a tree metric is an approximate  process. 

\subsection{De-noising with patches}
To improve the quality of the de-noised images, we run the experiment for {\em patches} of the image, instead of pixels. Moreover, we use Algorithm \ref{alg:twod} which implements not only a pruning step, but also computes the solution directly.
In this experiment (see Figure \ref{fig:results2} for a sample of the results), each patch (a grid of pixels) from the noisy image is a query point, and the dataset consists of available patches which we use as a substitute for a noisy patch. 

In our experiment, to build the dataset, we take one image from the Berkeley segmentation data set, then add noise to the right half of the image, and try to use the patches from the left half to denoise the right half. Each patch is of size $5 \times5$ pixels. We obtain $317\times236$ patches from the left half of the image and use it as the patch database. Then we apply Algorithm \ref{alg:twod} to denoise the image. In particular, for each noisy patch $q_n$ (out of $317\times237$ patches) in the right half of the image, we perform a linear scan to find the closest patch $p_i$ from the patch database, based on the following cost function:
\[
dist(q_n, p_i)+ \sum_{p_j\in neighbor(q_n)} \frac{dist(p_j, p_i)}{5}
\]
where $dist(p, q)$ is defined to be the sum of squares of the $l_2$ distances between the colors of corresponding pixels in the two patches.

After that, for each noisy patch we retrieve the closest patch from the patch database. Then for each noisy pixel $x$, we first identify all the noisy patches (there are at most 25 of them) that cover it. The denoised color of this pixel $x$ is simply the average of all the corresponding pixels in those noisy patches which cover $x$.

Since the nearest neighbor algorithm is implemented using a linear scan, it takes around 1 hour to denoise one image. One could also apply some more advanced techniques like locality sensitive hashing to find the closest patches with much faster running time.

\begin{table}[!h]
\centering
 \begin{tabular}{ccc}
  \includegraphics[width=43mm]{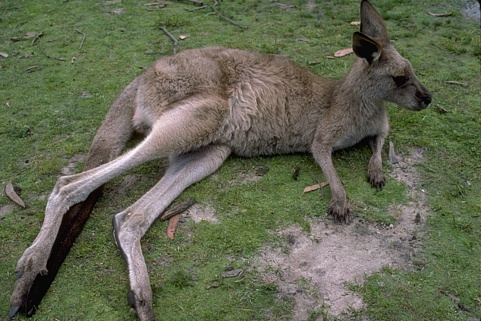}
&
  \includegraphics[width=43mm]{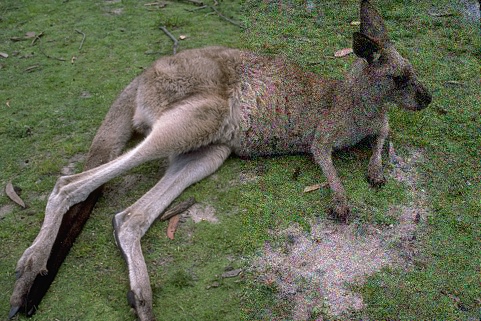}
&
  \includegraphics[width=43mm]{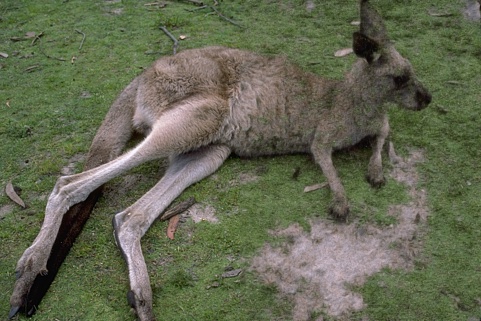}

\\
  \includegraphics[width=43mm]{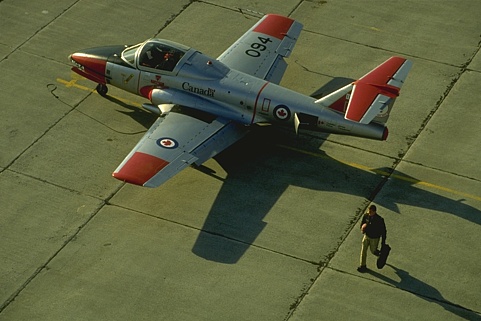}
&
  \includegraphics[width=43mm]{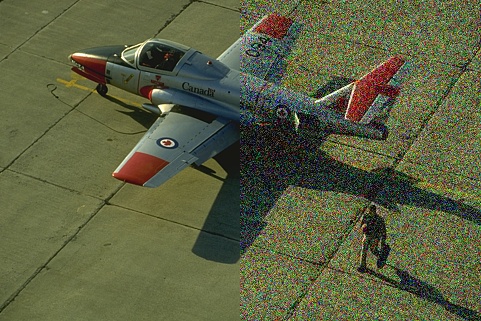}
&
  \includegraphics[width=43mm]{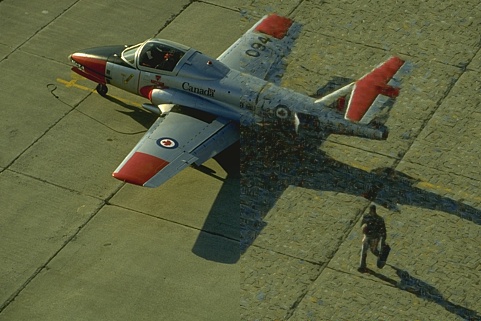}
\\
\end{tabular}
\caption{Two images from the Berkeley segmentation dataset \cite{berkeleyseg} (size $321*481$). The first column shows the original image; the second column shows the half noisy image; the third column shows the de-noised image using our algorithm for the patches.
}
\label{fig:results2}
\end{table}

\subparagraph*{Acknowledgements}
The authors would like to thank Pedro Felzenszwalb for formulating the Simultaneous Nearest Neighbor problem, as well as many helpful discussions about the experimental setup.



\section{$2r+1$ approximation}
\label{s:2r1}
Motivated by the importance of the $r$-sparse graphs in applications, in this section we focus on them and present another algorithm (besides INN) which solves the SNN problem for these graphs. We note that unlike INN, the algorithm presented in this section is not just a pruning step, but it solves the whole SNN problem.

For a graph $G=(Q,E)$ of pseudoarboricity $r$, let the mapping function be $f: E\rightarrow Q$, such that 
for every $e=(q_i,q_j)$, $f(e)=q_i$ or $f(e)=q_j$, and that for each $q_i\in Q$, $|C(q_i)|\leq r$, where $C(q_i)$ is defined as $\{e|f(e)=q_i\}$.

Once we have the mapping function $f$, we can run Algorithm \ref{alg:twod} to get an approximate solution. 
Although the naive implementation of this algorithm needs $O(rkn)$ running time, by using the aggregate nearest neighbor algorithm, it can be done much more efficiently.
We have the following lemma on the performance of this algorithm.

\begin{algorithm}[!h]
\caption{Algorithm for graph with pseudoarboricity $r$}
\label{alg:twod}
\vspace{0.2cm}
\textbf{Input} Query points $q_1,\cdots,q_k$, the input graph $G=(Q,E)$ with pseudoarboricity $r$\\
\textbf{Output} An Assignment $p_1,\cdots,p_k \in P$ \\
\vspace{-0.5cm}
\begin{algorithmic}[1]

\vspace{0.2cm}
\FOR{$i=1$ \TO $k$}
	\STATE{Assign $p_i \leftarrow 
\min_{p\in P} \dist(q_i,p) + 
\sum_{j:(q_i,q_j)\in C(q_j)} \frac{\dist(p, q_j)}{r+1} 
$}
\ENDFOR
\end{algorithmic}
\end{algorithm}

\begin{lemma}
If $G$ has pseudoarboricity $r$, the solution of Algorithm \ref{alg:twod} gives $2r+1$ approximation to the optimal solution. 
\end{lemma}
\begin{proof}
Denote the optimal solution as $P^*=\{p_1^*,\cdots,p_k^*\}$. We know the optimal cost is
\begin{align*}
\cost(Q,G,P^*)&=\sum_i \dist(q_i, p_i^*) + \sum_{(q_i,q_j)\in E} \dist(p_i^*,p_j^*)=\sum_i \left (
\dist(p_i^*, q_i)+\sum_{j: (q_i,q_j)\in C(q_j)}
\dist(p_i^*, p_j^*)
\right )
\end{align*}
\noindent Let $\sol$ be the solution reported by Algorithm \ref{alg:twod}. Then we have
\begin{align*}
&\cost(\sol)=
\sum_i \left (\dist(q_i, p_i) + \sum_{j:(q_i,q_j)\in C(q_j)} \dist(p_i,p_j)\right )\\
\leq&
\sum_i \left (\dist(q_i, p_i) + \sum_{j:(q_i,q_j)\in C(q_j)} \dist(p_i,q_j)+
\sum_{j:(q_i,q_j)\in C(q_j)} \dist(q_j,p_j)
\right )\quad\mbox{(by triangle inequality)} \\
\leq& 
\sum_i \left (\dist(q_i, p_i) + \sum_{j:(q_i,q_j)\in C(q_j)} \dist(p_i,q_j)\right )+
r \sum_j \dist(q_j,p_j)\quad\mbox{(by definition of pseudoarboricity)}
\\
=&
(r+1) \sum_i \dist(q_i, p_i) + \sum_{(q_i,q_j)\in C(q_j)} \dist(p_i,q_j)\\
\leq& 
(r+1)\sum_i \left (\dist(q_i, p_i^*) + \sum_{j:(q_i,q_j)\in C(q_j)} \frac{\dist(p_i^*,q_j)}{r+1}\right )\quad\mbox{(by the optimality of $p_i$ in the algorithm)}\\
\leq&
(r+1)\sum_i \left (\dist(q_i, p_i^*) + \sum_{j:(q_i,q_j)\in C(q_j)} \frac{\dist(p_i^*,p_j^*)+ \dist(p_j^*, q_j)}{r+1}\right )
\quad\mbox{(by triangle inequality)}\\
\leq& 
(r+1)\cost(Q,G,P^*)+ 
\sum_i \sum_{j: (q_i,q_j)\in C(q_j)} \dist(p_j^*, q_j)\\
\leq& 
(r+1) \cost(Q,G,P^*)+ 
r\sum_j  \dist(p_j^*, q_j)\quad\mbox{(by definition of pseudoarboricity)}\\
=&
\left (2r+1\right )\cost(Q,G,P^*)\qedhere
\end{align*}
\end{proof}



 \makeatletter\def\@@number#1{#1}\makeatother
 
\bibliographystyle{unsrt}

 \bibliography{biblio}

\begin{thebibliography}{10}

\bibitem{grant}
Pedro Felzenszwalb, William Freeman, Piotr Indyk, Robert Kleinberg, and Ramin
  Zabih.
\newblock Bigdata: F: Dka: Collaborative research: Structured nearest neighbor
  search in high dimensions.
\newblock {\em \verb|http://cs.brown.edu/~pff/SNN/|}, 2015.

\bibitem{bentley1975multidimensional}
Jon~Louis Bentley.
\newblock Multidimensional binary search trees used for associative searching.
\newblock {\em Communications of the ACM}, 18(9):509--517, 1975.

\bibitem{arya1998optimal}
Sunil Arya, David~M Mount, Nathan~S Netanyahu, Ruth Silverman, and Angela~Y Wu.
\newblock An optimal algorithm for approximate nearest neighbor searching fixed
  dimensions.
\newblock {\em Journal of the ACM (JACM)}, 45(6):891--923, 1998.

\bibitem{indyk1998approximate}
Piotr Indyk and Rajeev Motwani.
\newblock Approximate nearest neighbors: towards removing the curse of
  dimensionality.
\newblock In {\em Proceedings of the thirtieth annual ACM symposium on Theory
  of computing}, pages 604--613. ACM, 1998.

\bibitem{kushilevitz2000efficient}
Eyal Kushilevitz, Rafail Ostrovsky, and Yuval Rabani.
\newblock Efficient search for approximate nearest neighbor in high dimensional
  spaces.
\newblock {\em SIAM Journal on Computing}, 30(2):457--474, 2000.

\bibitem{krauthgamer2004navigating}
Robert Krauthgamer and James~R Lee.
\newblock Navigating nets: simple algorithms for proximity search.
\newblock In {\em Proceedings of the fifteenth annual ACM-SIAM symposium on
  Discrete algorithms}, pages 798--807. Society for Industrial and Applied
  Mathematics, 2004.

\bibitem{andoni2014beyond}
Alexandr Andoni, Piotr Indyk, Huy~L Nguyen, and Ilya Razenshteyn.
\newblock Beyond locality-sensitive hashing.
\newblock In {\em Proceedings of the Twenty-Fifth Annual ACM-SIAM Symposium on
  Discrete Algorithms}, pages 1018--1028. SIAM, 2014.

\bibitem{freeman2002example}
William~T Freeman, Thouis~R Jones, and Egon~C Pasztor.
\newblock Example-based super-resolution.
\newblock {\em Computer Graphics and Applications, IEEE}, 22(2):56--65, 2002.

\bibitem{boykov2001fast}
Yuri Boykov, Olga Veksler, and Ramin Zabih.
\newblock Fast approximate energy minimization via graph cuts.
\newblock {\em Pattern Analysis and Machine Intelligence, IEEE Transactions
  on}, 23(11):1222--1239, 2001.

\bibitem{barnes2009patchmatch}
Connelly Barnes, Eli Shechtman, Adam Finkelstein, and Dan Goldman.
\newblock Patchmatch: A randomized correspondence algorithm for structural
  image editing.
\newblock {\em ACM Transactions on Graphics-TOG}, 28(3):24, 2009.

\bibitem{kleinberg2002approximation}
Jon Kleinberg and Eva Tardos.
\newblock Approximation algorithms for classification problems with pairwise
  relationships: Metric labeling and markov random fields.
\newblock {\em Journal of the ACM (JACM)}, 49(5):616--639, 2002.

\bibitem{yiu2005aggregate}
Man~Lung Yiu, Nikos Mamoulis, and Dimitris Papadias.
\newblock Aggregate nearest neighbor queries in road networks.
\newblock {\em Knowledge and Data Engineering, IEEE Transactions on},
  17(6):820--833, 2005.

\bibitem{li2011flexible}
Yang Li, Feifei Li, Ke~Yi, Bin Yao, and Min Wang.
\newblock Flexible aggregate similarity search.
\newblock In {\em Proceedings of the 2011 ACM SIGMOD international conference
  on management of data}, pages 1009--1020. ACM, 2011.

\bibitem{li2011group}
Feifei Li, Bin Yao, and Piyush Kumar.
\newblock Group enclosing queries.
\newblock {\em Knowledge and Data Engineering, IEEE Transactions on},
  23(10):1526--1540, 2011.

\bibitem{agarwal2012nearest}
Pankaj~K Agarwal, Alon Efrat, and Wuzhou Zhang.
\newblock Nearest-neighbor searching under uncertainty.
\newblock In {\em Proceedings of the 32nd symposium on Principles of database
  systems}. ACM, 2012.

\bibitem{kopelowitz2012faster}
Tsvi Kopelowitz and Robert Krauthgamer.
\newblock Faster clustering via preprocessing.
\newblock {\em arXiv preprint arXiv:1208.5247}, 2012.

\bibitem{karzanov}
Alexander~V Karzanov.
\newblock Minimum 0-extensions of graph metrics.
\newblock {\em European Journal of Combinatorics}, 19(1):71--101, 1998.

\bibitem{FHRT}
Jittat Fakcharoenphol, Chris Harrelson, Satish Rao, and Kunal Talwar.
\newblock An improved approximation algorithm for the 0-extension problem.
\newblock In {\em Proceedings of the fourteenth annual ACM-SIAM symposium on
  Discrete algorithms}, pages 257--265. Society for Industrial and Applied
  Mathematics, 2003.

\bibitem{CKR}
Gruia Calinescu, Howard Karloff, and Yuval Rabani.
\newblock Approximation algorithms for the 0-extension problem.
\newblock {\em SIAM Journal on Computing}, 34(2):358--372, 2005.

\bibitem{archer}
Aaron Archer, Jittat Fakcharoenphol, Chris Harrelson, Robert Krauthgamer, Kunal
  Talwar, and {\'E}va Tardos.
\newblock Approximate classification via earthmover metrics.
\newblock In {\em Proceedings of the fifteenth annual ACM-SIAM symposium on
  Discrete algorithms}, pages 1079--1087. Society for Industrial and Applied
  Mathematics, 2004.

\bibitem{leenaor}
James~R Lee and Assaf Naor.
\newblock Metric decomposition, smooth measures, and clustering.
\newblock {\em Preprint}, 2004.

\bibitem{doublingdimension}
Anupam Gupta, Robert Krauthgamer, and James~R Lee.
\newblock Bounded geometries, fractals, and low-distortion embeddings.
\newblock In {\em Foundations of Computer Science, 2003. Proceedings. 44th
  Annual IEEE Symposium on}, pages 534--543. IEEE, 2003.

\bibitem{KKMR}
Howard Karloff, Subhash Khot, Aranyak Mehta, and Yuval Rabani.
\newblock On earthmover distance, metric labeling, and 0-extension.
\newblock {\em SIAM Journal on Computing}, 39(2):371--387, 2009.

\bibitem{treemetric}
Pedro~F Felzenszwalb, Gyula Pap, Eva Tardos, and Ramin Zabih.
\newblock Globally optimal pixel labeling algorithms for tree metrics.
\newblock In {\em Computer Vision and Pattern Recognition (CVPR), 2010 IEEE
  Conference on}, pages 3153--3160. IEEE, 2010.

\bibitem{earlyvision}
Pedro~F Felzenszwalb and Daniel~P Huttenlocher.
\newblock Efficient belief propagation for early vision.
\newblock {\em International journal of computer vision}, 70(1):41--54, 2006.

\bibitem{gco}
Yuri Boykov and Vladimir Kolmogorov.
\newblock An experimental comparison of min-cut/max-flow algorithms for energy
  minimization in vision.
\newblock {\em Pattern Analysis and Machine Intelligence, IEEE Transactions
  on}, 26(9):1124--1137, 2004.

\bibitem{berkeleyseg}
David~R Martin, Charless~C Fowlkes, and Jitendra Malik.
\newblock Learning to detect natural image boundaries using local brightness,
  color, and texture cues.
\newblock {\em Pattern Analysis and Machine Intelligence, IEEE Transactions
  on}, 26(5):530--549, 2004.

\end{thebibliography}


\end{document}